\documentclass{ifacconf}
\usepackage{dsfont}
\usepackage[utf8x]{inputenc} %%%%%% add x
\usepackage{amsmath}
\usepackage{mathrsfs}
\usepackage{amsfonts}
\usepackage{amssymb}
\usepackage{graphicx}
\usepackage{color}
\usepackage{lipsum}
\usepackage{enumerate}
\usepackage{todonotes}
\usepackage{multirow}
\usepackage{natbib}        % required for bibliography
\usepackage{soul}
\newcommand{\cc}[1]{{\color{black} #1}} % Reserved for Pablo's comments
\newcommand{\ccz}[1]{{\color{black} #1}} % Reserved for Carmen's comments
\newcommand{\norm}[1]{\left\lVert#1\right\rVert}

\newcommand{\minlamb}[1]{\lambda_{\tt{min}}(#1)}
\newcommand{\maxlamb}[1]{\lambda_{\tt{max}}(#1)}
\newcommand{\inv}[1]{#1^{-1}}

\def\rea{\mathds{R}}
\def\cmp{\mathds{C}}
\newtheorem{remark}{Remark}
\newtheorem{theorem}{Theorem}
\newenvironment{proof}{\textit{Proof:}}{\hfill$\blacksquare$}
\newtheorem{proposition}{Proposition}

\def\BibTeX{{\rm B\kern-.05em{\sc i\kern-.025em b}\kern-.08em
		T\kern-.1667em\lower.7ex\hbox{E}\kern-.125emX}}

\begin{document}
\begin{frontmatter}

\title{Dead-zone compensation via passivity-based control for a class of mechanical systems}

\thanks[footnoteinfo]{The work of C. Chan-Zheng is sponsored by the University of Costa Rica.}
\author[First]{Carmen Chan-Zheng} 
\author[Second]{Pablo Borja} 
\author[First]{Jacquelien M.A. Scherpen}

\address[First]{ Jan C. Willems Center for Systems and Control, and Engineering and Technology institute Groningen (ENTEG), Faculty of Science and Engineering at the University of Groningen,  9747 AG Groningen, The Netherlands (email: c.chan.zheng@rug.nl, j.m.a.scherpen@rug.nl).}
\address[Second]{  School of Engineering, Computing and Mathematics, University of Plymouth, Plymouth, Devon PL4 8AA, United Kingdom  (email: pablo.borjarosales@plymouth.ac.uk).}

\begin{abstract}                % Abstract of not more than 250 words.
	This manuscript introduces a passivity-based control methodology for fully-actuated mechanical systems with symmetric or asymmetric dead-zones. To this end, we find a smooth approximation of the inverse of the function that describes such a nonlinearity. Then, we propose an energy and damping injection approach~---~based on the PI-PBC technique~---~that compensates for the dead-zone. Moreover, we provide an analysis of the performance of the proposed controller near the equilibrium. We conclude this paper by experimentally validating the results on a two degrees-of-freedom planar manipulator.
\end{abstract}

\begin{keyword}
damping injection, energy shaping, passivity, pid, tuning.
\end{keyword}

\end{frontmatter}
%===============================================================================

	\section{Introduction}

The implementation of approaches for controlling mechanical systems is often hindered by unmodeled nonlinear phenomena such as saturation, dry friction, or asymmetry of the motors. The dead-zones are particularly ubiquitous in servomechanisms (or actuators in general), and their presence severely affects the performance. For example, steady-state errors occur when the nonzero control input falls within the dead-zone region. To handle this nonlinearity, the authors in \cite{na2018adaptive} provide a compilation of recent results on adaptive control to compensate for asymmetric and symmetric dead-zones. Other methodologies to deal with dead-zones are fuzzy logic control (see \cite{woo1997deadzone,betancor2014deadzone}) or neural networks (see \cite{selmic2000deadzone}). In this manuscript, we focus on providing dead-zones compensation based on the passivity-based approach (PBC).

The PBC strategies offer a constructive approach for stabilizing a large class of complex physical systems where the exchange of energy between the plant and the environment plays a central role (see \cite{vanderSchaft2017,ortega2013passivity}). Customarily, these techniques follow two steps: i) the energy shaping process, which guarantees that the closed-loop system has a stable equilibrium at the desired configuration; and ii) the damping injection step, which ensures asymptotic stability properties for the desired equilibrium point. As a consequence of these steps, the PBC gains are associated with the physical quantities of the closed-loop system, i.e., energy and damping, endowing the controller with physical intuition. However, most of these techniques overlook the dead-zone resulting in steady-state errors as reported in \cite{chan2021tuning,chan2022tuning}. Albeit steady-state errors can be solved by adding an integral action to the PBC approach~---~as described in \cite{chan2022integral,ortega2012robust,dirksz2012power}~---~its implementation may be hampered due to the increase in complexity of the tuning process and unexpected behaviours such as wind-up, backslash or instability.
Other PBC techniques that deal with dead-zones are found in \cite{mizumoto2012control,jung2022synchronous}. The former proposes a passivity-based adaptive control approach, while the latter, presents a passivity-based sliding mode control.

In this paper, we design a PBC approach with dead-zone compensation based on a PI-PBC strategy described in \cite{ortega2021pid, borja2020new}. The advantage of this technique is that the control design process is reduced to select the gains properly to achieve a particular performance.
The proposed control technique is suitable for fully-actuated mechanical systems described in the port-Hamiltonian (pH) setting (see \cite{duindam2009modeling}). The main advantage is that this framework highlights the role of the interconnection structure, dissipation, and energy play in the system behaviour. Additionally, the proposed PBC approach is suitable to compensate for both symmetric and asymmetric dead-zones, and its stability is proven by invoking passivity properties. Note that by avoiding dynamical extension~---~as found in adaptive or integral control methodologies~---~we ease the practical implementation of the controller as the tuning process is simplified. Our main contributions are summarized as follows:
\begin{itemize}
	\item A PBC technique based on a PI-PBC approach that deals with symmetric and asymmetric dead-zones.
	\item A performance analysis in a vicinity of the equilibrium of the closed loop. 
\end{itemize}

The remainder of this paper is structured as follows: in Section \ref{prel} we provide the theoretical background and formulate the problem under study.  In Section \ref{main}, we present the main results of this paper, i.e., the PBC technique with dead-zone compensation. Section \ref{performance} provides some remarks on the performance of the closed-loop system. In Section \ref{exp_res}, we show experimental results obtained from a 2 degrees-of freedom (DoF) planar manipulator. We finalize this manuscript with some concluding remarks and future research in Section \ref{conclusions}.

\textbf{Notation}: We denote the $n\times n$ identity matrix as $I_n$ and the $n\times m$ matrix of zeros as $0_{n\times m}$. For a given smooth function $f:\rea^n\to \rea$, we define the differential operator $\nabla_x f:=\frac{\partial f}{\partial x}$ which is a column vector, and $\nabla^2_x f:=\frac{\partial^2 f}{\partial x^2}$. For a smooth mapping $F:\rea^n\to\rea^m$, we define the $ij-$element of its $n\times m$ Jacobian matrix as $(\nabla_x F)_{ij}:=\frac{\partial F_i}{\partial x_j}$. When clear from the context the subindex in $\nabla$ is omitted. For a given vector $x\in\rea^{n}$, we say that $A$ is \textit{positive definite (semi-definite)}, denoted as $A\succ0$  ($A\succeq0$), if $A=A^{\top}$ and $x^{\top}Ax>0$ ($x^{\top}Ax\geq0$) for all $x\in \rea^{n}-\{0_{n} \}$ ($\rea^{n}$). For a given vector $x\in\rea^n$ , we denote the Euclidean norm as $\norm{x}$.
Given the distinguished element $x_\star\in\rea^n$, we define the constant matrix $B_\star:=B(x_\star)\in\rea^{n\times m}$. We denote $\rea_+$ as the set of positive real numbers and $\rea_{\geq0}$ as the set $\rea_+\cup \{0\}$. Let $x\in\rea^{n}$ and $y\in\rea^m$, we define $col(x,y):=[x^\top y^\top]^\top$. We denote $e_i$ as the $i^{th}$ element of the canonical basis of $\rea^n$. All the functions considered in this manuscript are assumed to be (at least) twice continuously differentiable. 

\textbf{Caveat}: when possible, we omit the arguments to simplify the notation.
\section{Preliminaries and Problem setting }\label{prel}
This section describes the pH representation of the class of mechanical systems considered throughout this paper. Moreover, we characterize the dead-zone phenomenon and we revisit a PI-PBC technique, which we compare against our approach in Section \ref{exp_res}. We conclude this section with the problem formulation.
\subsection{Description of a class of mechanical systems}
Consider a pH mechanical system of the form
\cc{\begin{equation}\label{sysmec}
	\arraycolsep=1pt \def\arraystretch{1.2}
	\begin{array}{rcl}
		\dot{x} &=& \begin{bmatrix}
			0_{n\times n}&I_n\\
			-I_n &-{D(x)}\end{bmatrix}\begin{bmatrix}
			\nabla_q{H}(x) \\ \nabla_{p}{H}(x)
		\end{bmatrix}+\begin{bmatrix}
			0_{n\times n} \\ G
		\end{bmatrix}u+\begin{bmatrix}
		0_n\\\beta
	\end{bmatrix}\\
		y&=&G^\top\inv{M}{{(q)p}}
	\end{array}
\end{equation}}
where $x:=col(q,p)$ with ${q},{p} \in \rea^{n}$ being the generalized positions and momenta vectors, respectively; ${H:\rea^n\times\rea^n\to \rea}$ is the Hamiltonian of the system defined as
\begin{equation}\label{OLHam}
	{H}(x)=\displaystyle\frac{1}{2}{p}^\top \inv{M}({q}) {p}+U(q);
\end{equation}
the potential energy of the system is denoted with ${U:\rea^n\to\rea}$; ${M:\rea^n\to\rea^{n\times n}}$ corresponds to the mass-inertia matrix verifying $M(q)\succ 0_{n\times n}$; $D:\rea^n\times \rea^n\to\rea^{n\times n}$ represents the natural damping satisfying ${D(q,p)\succeq 0_{n\times n}}$; $u,y \in \rea^{n}$ is the input and output vector, respectively; $G\in\rea^{n\times n}$ is the non-singular input matrix\footnote{This matrix captures the conversion of the physical domain of the actuator (e.g., electrical, pneumatic, or hydraulic) to the mechanical domain. For simplicity, we have defined it as diagonal matrix.}  defined as $G:=diag\{g_1,\hdots,g_n\}$ with $g_i\in\rea$; and $\beta\in\rea^n$ is a constant vector that captures a particular behaviour of the system explained in the next section.

\subsection{Dead-zone characterization}
\begin{figure}[t]
	\centering
	\includegraphics[width=0.7\columnwidth]{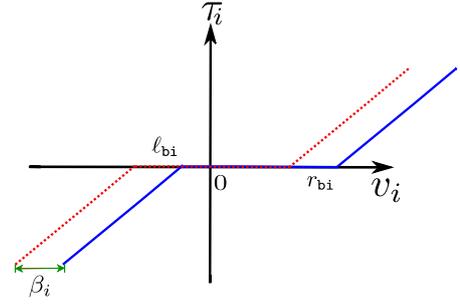}
	\caption{\centering Symmetric dead-zone (dotted red line) and asymmetric dead-zone (solid blue line), where $\beta$ represents the shift with respect to the vertical axis of the asymmetric dead-zone.}\label{dz}
\end{figure}
\cc{A particular phenomenon inherent in any actuator is the presence of dead-zones, i.e., the range of input control values where the output (force or torque) is zero. A simple dead-zone model is given by
\begin{equation}\label{dzcharacteristic}
		\tau_i:=\begin{cases}
				v_i-r_{\tt{bi}},\qquad &v_i> r_{\tt{bi}}\\
				0,\qquad &l_{bi}\leq v_i\leq r_{\tt{bi}}\\
				v_i-\ell_{\tt{bi}},\qquad &v_i< \ell_{\tt{bi}},
			\end{cases}
\end{equation}
with $i\in \{1,\hdots,n\}$; $r_{\tt{bi}}> 0$; $\ell_{\tt{bi}}< 0$; the input $v=col(v_1,\hdots,v_n)$ with $v_i\in\rea$; and the output $\tau:=col(\tau_1,\hdots,\tau_n)$ with $\tau_i \in \rea$. The graphical representation of \eqref{dzcharacteristic} is given in Fig.~\ref{dz}. Note that $\beta_i \in \rea$ represents an offset term that transforms the dead-zone from symmetric to asymmetric. In this manuscript, we represent this offset with the vector $\beta:=col(\beta_1,\hdots,\beta_n)$ considered in \eqref{sysmec}.}%\footnote{For symmetric dead-zones, we have that $\beta=0_n$}}

\subsection{A PI-PBC strategy}
The gains of the PI-PBC strategy~--~described in \cite{ortega2021pid,borja2020new}~--~admit a physical interpretation which endows the tuning process with more intuition. Moreover, its practical implementation is eased as the control design process is reduced to select properly to achieve a certain performance. However, this control technique does not account for the dead-zone effect resulting in steady-state errors as reported in \cite{chan2021tuning,chan2022tuning}. In this paper, we aim to extend such a methodology by adapting it to deal with such a nonlinearity. For the sake of completeness, we summarize the PI-PBC approach in Proposition \ref{prop1}.
\begin{proposition}[PI-PBC]\label{prop1}
	Consider the pH system \eqref{sysmec} with $\beta=0_n$; the desired equilibrium  $x_\star:=(q_\star,0_n)$; the constant matrices $K_P,K_I\in\rea^{n\times n}$ verifying $K_I\succ0$ and $K_P\succeq 0$; the control law 
	\begin{equation}\label{pi}
		u=u_{\tt{pi}}:=-\inv{G}(K_P\dot{q}+K_I(q-{q}_\star)-\nabla_q U(q))
	\end{equation}
	and assume that 
	$$H_{\tt{pi}}(x):=\dfrac{1}{2}p^\top \inv{M}(q)p+\dfrac{1}{2}(q-\tilde{q})^\top K_I (q-\tilde{q})$$
	verifies
	$$x_\star =\arg \min H_{\tt{pi}} (x).$$
	It follows that the closed-loop system has a globally asymptotically stable equilibrium at $x_\star$ with Lyapunov candidate $H_{\tt{pi}}(x)$.		\hfill $\blacksquare$
\end{proposition}

\subsection{Problem formulation}
Now, we formulate the problem under study as follows\\
\\
\textit{``Propose a passivity-based control strategy suitable to compensate dead-zones, symmetric and asymmetric, while stabilizing \eqref{sysmec} at the desired equilibrium $x_\star$."} 

\section{Main Results}\label{main}

The dead-zone may severely impact the closed-loop performance since the nonzero control input falls within the dead-zone region resulting in steady-state errors. In this section, we describe a PBC strategy to handle this particular nonlinearity; the main advantage of employing a PBC methodology is that the energy function of the closed-loop can be chosen as a Lyapunov candidate, simplifying the stability proof. 

The proposed PBC methodology in this paper, i.e., $u_{\tt{pidz}}$, is an extension of the standard PI-PBC \eqref{pi}. For this extension, we employ a scheme based on the traditional method (see \cite{cho1998convergence,rubio2013proportional}), which consists in: i) defining the inverse of the dead-zone of Fig.~\ref{dz}, whose graphical representation is visualized in Fig.\ref{dz_inv}a, and then, ii) designing the controller $u_{\tt{pidz}}$.
\begin{figure}[t]
	\centering
	\includegraphics[width=0.8\columnwidth]{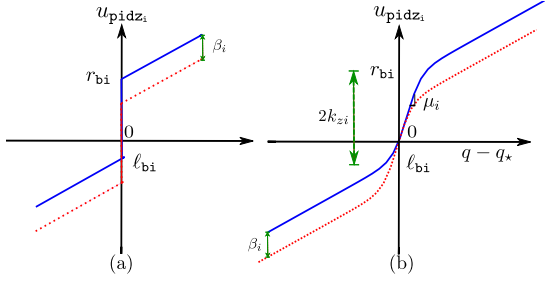}
	\caption{\centering (a) Inverse of Fig.~\ref{dz}. (b) Smooth inverse approximation.}\label{dz_inv}
\end{figure}

In particular, we consider the following structure for the controller ${u_{\tt{pidz}}}$
\begin{equation}\label{ucompdz}
 	u_{\tt{pidz}}:=u_{\tt{dz}}+u_{\tt{pi}}.
\end{equation}
To formulate the dead-zone compensator part $u_{\tt{dz}}$, we employ a smooth approximation of Fig.~\ref{dz_inv}a, which is given in Fig.~\ref{dz_inv}b. The full design of the proposed PBC is described in Theorem \ref{thm1}.

	\begin{theorem}\label{thm1}
	Consider the fully-actuated mechanical system \eqref{sysmec}; the desired equilibrium $x_\star:=(q_\star,0_n)$; and the control law \eqref{ucompdz} with $u_{\tt{pi}}$ as defined in \eqref{pi} and
	\begin{equation}\label{controllaw1}
		\arraycolsep=1pt \def\arraystretch{1.2}
		\begin{array}{rcl}
			u_{\tt{dz}}:=-\inv{G}(K_Z\sum_{i=1}^{n}e_i\tanh(\mu_i \tilde{q}_i)+\beta),
		\end{array}
	\end{equation}
	where $K_P, K_I \in \rea^{n\times n}$ verifies $K_P,K_I\succ 0_{n\times n}$; $i\in\{1,\hdots,n\}$; $K_Z:=diag\{k_{z1},\hdots,k_{zn}\}$ with $k_{zi} \in \rea_{+}$; $\beta:=col(\beta_{1},\hdots,\beta_{n})$; and $\mu:=diag(\mu_{1},\hdots,\mu_{n})$ with $\mu_i \in \rea_{+}$. \cc{Then, the closed loop \eqref{sysmec}-\eqref{ucompdz} has a \textit{global asymptotic stable} equilibrium at $x_\star$ with Lyapunov candidate
	\begin{equation}\label{Hd}
		\arraycolsep=1pt \def\arraystretch{1.2}
		\begin{array}{rcl}
			{H_d}(x)=\displaystyle\frac{1}{2}{p}^\top \inv{M}({q}) {p}&+&\dfrac{1}{2}\tilde{q}^\top K_I \tilde{q}+K_Z\sum_{i=1}^{m}\dfrac{\ln(\cosh(\mu_i\tilde{q}_i))}{\mu_i}
		\end{array}
	\end{equation}	
	where $\tilde{q}:=q-q_\star$.}
\end{theorem}
\begin{proof}
	Substituting \eqref{ucompdz} in \eqref{sysmec}, we notice that the closed loop preserves mechanical structure, i.e.,
	\begin{equation}\label{tg}
		\arraycolsep=1pt \def\arraystretch{1.2}
		\begin{array}{rcl}
			\begin{bmatrix}
				\dot{{q}}\\\dot{{p}}
			\end{bmatrix} &=& \begin{bmatrix}
				0_{n\times n}&I_n\\
				-I_n &-{D}({q,p})-K_P\end{bmatrix}\nabla H_d(x)		
		\end{array}
	\end{equation}
	where $H_d(x)$ is the desired Hamiltonian defined in \eqref{Hd}.
	
	Then, considering $H_d(x)$ as the Lyapunov candidate, it follows that
	\begin{eqnarray}\label{jac}
		\left(\nabla_q {H_d}\right)_\star=0_n,~
		\left(\nabla_p H_d\right)_\star=0_n.
	\end{eqnarray}
	Moreover, we have that 
	\begin{equation}\label{hessian}
		\arraycolsep=1pt \def\arraystretch{1.2}
		\begin{array}{rcl}
			\left(\nabla_q^2 H_d\right)_\star&=&K_I+\mu K_Z\succ0_{n\times n},\\
			\left(\nabla_{qp} H_d\right)_\star&=&\left(\nabla_{pq} H_d \right)_\star=0_{n\times n},\\
			\left(\nabla_p^2 H_d\right)_\star&=&\inv{M}(q_\star)\succ0_{
				n\times n}.\\
		\end{array}
	\end{equation}
	which guarantees that 
	\begin{equation*}
		\left(\nabla_x^2 H_d\right)_\star\succ0_{2n\times 2n}.
	\end{equation*}
	Therefore, from \eqref{jac}-\eqref{hessian}, we have that $x_\star$ is a stable equilibrium of the closed loop \eqref{sysmec}-\eqref{ucompdz} (for further details, see Proposition 7.2.8 from \cite{vanderSchaft2017}).
	
	To prove asymptotic stability of $x_\star$, we invoke LaSalle's invariance principle (see \cite{khalil2002nonlinear}). Note we have that
	\begin{equation*}
		\arraycolsep=1pt \def\arraystretch{1.2}
		\begin{array}{rcl}
			\dot{H}_d&=& -\nabla_p^\top H_d (D+K_P)\nabla_p H_d\leq 0.	
		\end{array}
	\end{equation*}
	\ccz{Recall that, from \eqref{tg}, we have 
	\begin{equation*}\arraycolsep=1pt \def\arraystretch{1.2}
		\begin{array}{rcl}
			\dot{q}&=&\nabla_{p}H_d=\inv{M}(q)p\\
			\dot{p}&=&-\nabla_q H_d-(D(q,p)+K_P)\nabla_p H_d,
	 	\end{array}
	\end{equation*}	
	with $\nabla_q H_d= \dfrac{1}{2}\nabla_q\left(p^\top \inv{M}p\right)+K_I\tilde{q}+K_Z\displaystyle\sum^n_{i=1}e_i\tanh(\mu_i \tilde{q}_i)$. Accordingly, we have the following chain of implications}
	\begin{equation*}
		\arraycolsep=1pt \def\arraystretch{1.2}
		\begin{array}{rcl}
			&\dot{H}_d&\equiv0	\iff \dot{q}=0_n \iff p=0_n \implies \dot{p}=0_n\\
			&\iff& K_I\tilde{q}+K_Z\sum^n_{i=1}e_i\tanh(\mu_i \tilde{q}_i)=0_n\\
			&\iff& q=q_\star.
		\end{array}
	\end{equation*}
	Thus, the closed loop \eqref{sysmec}-\eqref{ucompdz} has an asymptotically stable equilibrium point at $x_\star$.
	We conclude the proof by proving the global properties of the equilibrium point; note that the Lyapunov candidate is radially unbounded since ${H}_d\to\infty$ as $\norm{q}\to\infty$ and $\norm{p}\to\infty$.
\end{proof}

\begin{remark}
	We employ $\tilde{q}_i$ as the feedback variable in $u_{\tt{dz}}$ because we need to ensure a nonzero input while the system has  not reached the equilibrium, i.e., $\tilde{q}_i \neq 0$. 
\end{remark}
\begin{remark}\label{r1}
	Regarding the selection of the parameters of \eqref{controllaw1}, \cc{the term $\beta$ corresponds to the offset that shifts the dead-zone centered in zero (see Fig.~\ref{dz}; in particular,  $\beta=0_n$ corresponds to a symmetric dead-zone). Thus, a proper characterization of this parameter is required to implement of the controller \eqref{controllaw1}.} Additionally, the constant $k_{zi}$ is associated with the length of the dead-zone, i.e., $k_{zi}:=\dfrac{r_{\tt{bi}}-\ell_{\tt{bi}}}{2}$ (see Fig~\ref{dz_inv}b). Furthermore, note that selecting $\mu_i>>0$ improves the approximation of Fig~\ref{dz_inv}b to Fig~\ref{dz_inv}a.
\end{remark}

In the scenario that the dead-zone is not properly characterized, the controller still improves the steady-state error. For example, if the length of the dead-zone is underestimated, a steady-state error remains but is reduced. On the other hand, an overshoot in the transient response may occur if the length is overestimated. The latter behaviour is explained in the next section, where we analyze the closed-loop behaviour near the equilibrium.

\section{Performance near the equilibrium}\label{performance}
In this section, we discuss the behaviour of the closed-loop \eqref{sysmec}-\eqref{ucompdz} near the equilibrium $x_\star$. To this end, we follow the same procedure as in \cite{chan2021tuning,chan2022tuning}, which consists of linearizing the closed-loop around the desired equilibrium, and then finding a transformation such that the linearized matrix has a \textit{saddle point form}. The latter structure has been shown effective in analyzing the performance of the system near the equilibrium.

We first introduce the linearized vector $\hat{x}:=x-x_\star$. Then, the linearized dynamics of the closed loop \eqref{sysmec}-\eqref{ucompdz} corresponds to
\begin{equation}\label{lin}
	\dot{\hat{x}}=-\begin{bmatrix}
		0_{n\times n}&-\inv{M}_\star\\
		K_I+\mu K_z&(D+K_P)\inv{M}_\star
	\end{bmatrix}\hat{x}.
\end{equation}
To obtain the saddle point matrix form of \eqref{lin}, we introduce the similarity transformation matrix 
$$T:=\begin{bmatrix}
	0_{n\times n}&\phi_M\\
	\phi_P&0_{n\times n},
\end{bmatrix}$$
where $\phi_M,\phi_P\in \rea^{n\times n}$ are upper triangular matrices obtained from the Cholesky decomposition (see \cite{horn2012matrix})
$$\inv{M}_\star=\phi_M^\top\phi_M,~K_I+\mu K_z=\phi_P^\top \phi_P.$$
Therefore, the dynamics in the new coordinates $\hat{z}:=T\hat{x}$ correspond to
\begin{equation}\label{saddle}
	\dot{\hat{z}}:=-\mathcal{N}\hat{z}, ~\mathcal{N}:=\begin{bmatrix}
		\phi_M(D+K_P)\phi_M^\top&\phi_M\phi_P^\top\\
		-\phi_P\phi_M^\top&0_{n\times n}
	\end{bmatrix},
\end{equation}
where $\mathcal{N}$ is a class of saddle point matrices (see \cite{benzi2006eigenvalues}). 

Subsequently, we study the performance of the system around the equilibrium by analyzing the eigenvalue problem associated with \eqref{saddle}, i.e., $\mathcal{N}w=\lambda w$ with $\lambda\in\cmp$ being an eigenvalue of $\mathcal{N}$ and ${w:=col(w_1,w_2)}$ being its corresponding eigenvector with $w_1,w_2 \in \cmp^n$. Rewriting the eigenvalue problem, we get that
\begin{equation}\label{saddle2}
	\arraycolsep=1pt \def\arraystretch{1.2}
	\begin{array}{rc}
		\phi_M(D+K_P)\phi_M^\top w_1+\phi_M\phi_P^\top w_2&=\lambda w_1\\
		-\phi_P\phi_M^\top w_1&=\lambda w_2
	\end{array}			
\end{equation}
Then, from \eqref{saddle2}, we obtain 
\begin{equation}\label{eig_pro}
	\lambda^2 - \dfrac{w_1^* \phi_M\mathcal{R}\phi_M^\top w_1}{\norm{w_1}^2}\lambda+\dfrac{v^*\phi_M\mathcal{P}\phi_M^\top w_1}{\norm{w_1}^2}=0,
\end{equation}
with $\mathcal{R}:= D+K_P$ and $\mathcal{P}:=K_I+\mu K_Z$. Then, the solution of \eqref{eig_pro} is given by
\begin{equation}\label{saddle3}
	\arraycolsep=1pt \def\arraystretch{1.2}
	\begin{array}{rcl}
		\lambda&=& \dfrac{1}{2} \dfrac{w_1^* \phi_M\mathcal{R}\phi_M^\top w_1}{\norm{w_1}^2}\\
		&\pm&\sqrt{\left(\dfrac{w_1^*\phi_M\mathcal{R}\phi_M^\top w_1}{\norm{w_1}^2}\right)^2-4\dfrac{w_1^*\phi_M \mathcal{P}\phi_M^\top w_1}{\norm{w_1}^2}}.
	\end{array}
\end{equation}

Note that \eqref{saddle3} provides a solution for every eigenvalue of $\mathcal{N}$ in terms of the gains $K_P,K_I, K_Z,\mu$.  By closed inspection of \eqref{saddle3}, we obtain information such as rise-time, damping ratio, or oscillations as shown in \cite{chan2021tuning,chan2022tuning}. A direct result from \eqref{saddle3} is summarized in Theorem 2.
\begin{theorem}\label{thm2}
	If
	\begin{equation}\label{tuningrule}
		4\maxlamb{\mathcal{P}}\maxlamb{M_\star}\leq  \minlamb{\mathcal{R}}^2,
	\end{equation}
	then, the spectrum of $\mathcal{N}$ is real and positive.
\end{theorem}
\begin{proof}
	The proof follows similarly as Proposition 1 from \cite{chan2021tuning}. Note that $\lambda\in\rea_{+}$ if and only if the discriminant is nonnegative,  i.e.,
	\begin{equation*}
		4\dfrac{w_1^* (\phi_M \mathcal{P}\phi_M^\top)w_1}{\norm{w_1}^2}\leq 		\left(\dfrac{w_1^* (\phi_M\mathcal{R}\phi_M^\top)w_1}{\norm{w_1}^2}\right)^2.
	\end{equation*}
	Consider $\eta:=\phi_M^\top w_1$; then,  we have that
	\begin{equation*}
		\arraycolsep=1pt \def\arraystretch{2}
		\begin{array}{rc}
			4\dfrac{\eta^\star\mathcal{P}\eta}{\eta^\star M_\star \eta}&\leq \left(\dfrac{\eta^\star \mathcal{R} \eta}{\eta^\star M_\star \eta}\right)^2\\
			\implies 4\dfrac{(\eta^\star\mathcal{P}\eta)(\eta^\star M_\star \eta)}{(\eta^\star\eta)^2}&\leq \left(\dfrac{\eta^\star \mathcal{R} \eta}{(\eta^\star\eta)}\right)^2
		\end{array}
	\end{equation*}
	Rewriting, we have the following chain of inequalities
	\begin{equation*}
		\arraycolsep=1pt \def\arraystretch{1.2}
		\begin{array}{rcl}
			4\dfrac{(\eta^\star\mathcal{P}\eta)(\eta^\star M_\star \eta)}{(\eta^\star\eta)^2}&\leq& 4\maxlamb{\mathcal{P}}\maxlamb{M_\star} \\
			&\leq& \minlamb{\mathcal{R}}^2\leq\left(\dfrac{\eta^\star \mathcal{R} \eta}{\eta^\star\eta}\right)^2
		\end{array}
	\end{equation*}
	If \eqref{tuningrule} holds, then, the spectrum of $\mathcal{N}$ is real. We conclude this proof by noting that $\lambda\in\rea_{+}$ since $\mathcal{R}\succ 0_{n\times n}$.
\end{proof}

Recall that the imaginary part of the eigenvalues partly characterizes the oscillations. Therefore, condition \eqref{tuningrule} ensures that the spectrum contain no imaginary part, and consequently, it guarantees that there are no oscillations (and overshoot) in the transient response.
\begin{remark}
	Theorem \ref{thm2} is similar to the results in \cite{chan2021tuning,chan2022tuning}. However, we underscore that these results are obtained from different PBC strategies. 
\end{remark}
\begin{remark}\label{r2}
	Note that by selecting a larger $\mu_i$ (\ccz{or equivalently, setting the slope steeper in Fig.~\ref{dz_inv}b}) or overestimating $k_{zi}$, it follows that the transient response may exhibit an overshoot as condition \eqref{tuningrule} no longer holds. 
\end{remark}
\begin{remark}
	One of the main advantages of employing a PBC strategy is that its gains are associated with the physical quantities of the system. In particular, for the control \eqref{ucompdz}, the parameters $K_I,K_Z$, and $\mu$ are associated with the potential energy shaping, while $K_P$ is related to the dissipation. 
\end{remark}

\section{Case study: A 2-DoF Planar manipulator}\label{exp_res}
In this section, we demonstrate the effectiveness of the proposed controller in compensating for dead-zones. To this end, we employ a 2-DoF planar manipulator as shown in Fig.~\ref{quanser} with an asymmetric dead-zone. We refer the reader to \cite{quanser} for further details on this experimental setup. 
\begin{figure}[t]
	\centering
	\includegraphics[width=0.5\columnwidth]{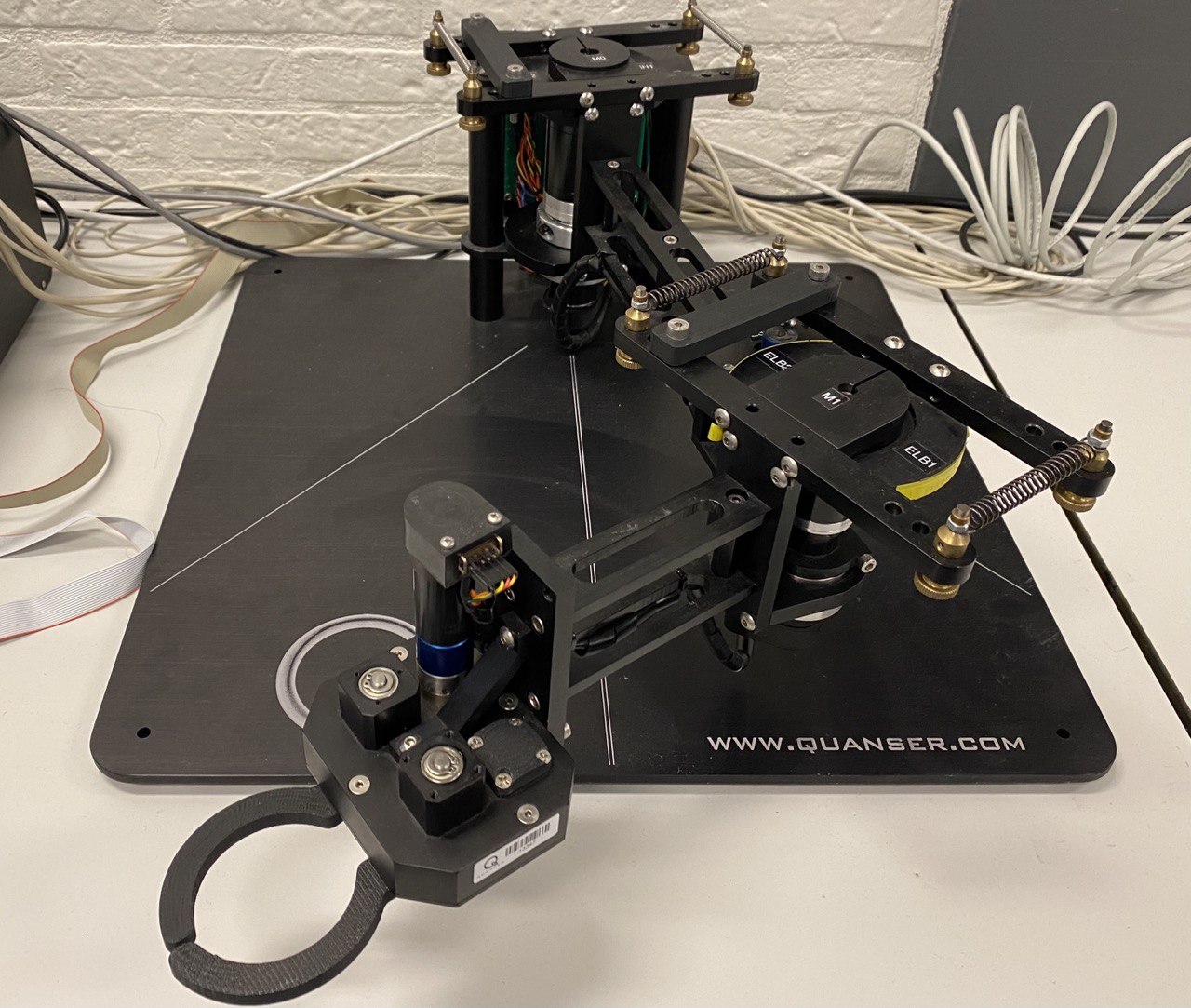}
	\caption{\centering A 2-DoF Quanser planar manipulator.}\label{quanser}
\end{figure}
The mathematical model of this system is given in \eqref{sysmec} with $n=2$. Moreover, the position and momenta of each link are denoted with $q_i$ and $p_i$ with $i=\{1,2\}$, respectively; $U(q)=0$; ${G=diag\{1,0.6\}}$; $D=diag\{1.5964, 0.6971\}$ and 
\begin{equation*}
	\def\arraystretch{1.2}
	\begin{array}{rl}
		M(q_2)&:=\begin{bmatrix}
			a_1+a_2+2b\cos(q_2)&a_2+b\cos(q_2)\\
			a_2+b\cos(q_2)&a_2
		\end{bmatrix},
	\end{array}
\end{equation*} 
where $a_1=0.1547,~a_2=0.0111$, $b=0.0168$.

First, we compare $u_{\tt{pidz}}$ in \eqref{ucompdz} against the standard PI-PBC in \eqref{pi}. The gains for each controller are shown in Table \ref{gain1}. The parameters of \eqref{ucompdz} are chosen with the assumption that the dead-zone is symmetric (i.e., $\beta=0_n$). 
\begin{table}[h!]
	\centering
	\caption{\centering PBC gains}\label{gain1}
	\begin{tabular}{ccc}
		\hline
		\multicolumn{1}{l}{} & \multicolumn{1}{l}{PI-PBC} & \multicolumn{1}{l}{PI-PBC with DZ compensation} \\ \hline
		$K_P$                   & $diag\{1.5,1\}$   & $diag\{1.5,1\} $                                        \\
		$K_I$                   & $diag\{5,3\}$       & $diag\{5,3\}$                                      \\
		$K_Z$                   &                  -       &    $diag\{0.13,0.35\}$                                   \\
		$\mu$                   &                        -    &     10                               \\
		$\beta$                 &                         -  &           $0_2$                          \\ \hline
	\end{tabular}
\end{table}

Next, we select five different configurations to show the efficacy of our controller for dealing with dead-zones. The results are shown in Table \ref{steady} and Fig.~\ref{res}, where the proposed controller \eqref{ucompdz} improves the steady-state error for all cases in comparison to the results of the PI-PBC from \eqref{pi}.	
\begin{table}[h!]
	\centering
	\caption{\centering Steady-state error per controller}\label{steady}
	\begin{tabular}{cccc}
		\hline
		Case & Position        & $u_{\tt{pi}}$(\%L1/ \%L2)   & $u_{\tt{pidz}}$(\%L1/ \%L2)  \\ \hline
		a    & {[}0.6,0.8{]}   &       2.17/6.63 & 1.67/1.75             \\
		b    & {[}-0.6,-0.8{]} &       4.17/4.75  &    2.83/1.375         \\
		c    & {[}-0.4,0.7{]}  &       5.25/6.86  & 3.75 / 3.07         \\
		d    & {[}0.4,-0.7{]}  &       2.75/5  & 2.5/   2.51        \\
		e    & {[}0.5,-0.5{]}  &       1.80/5.6 & 1.54/2.54            \\\hline
	\end{tabular}
\end{table}
\begin{figure}[h!]
	\centering
	\includegraphics[width=0.85\columnwidth]{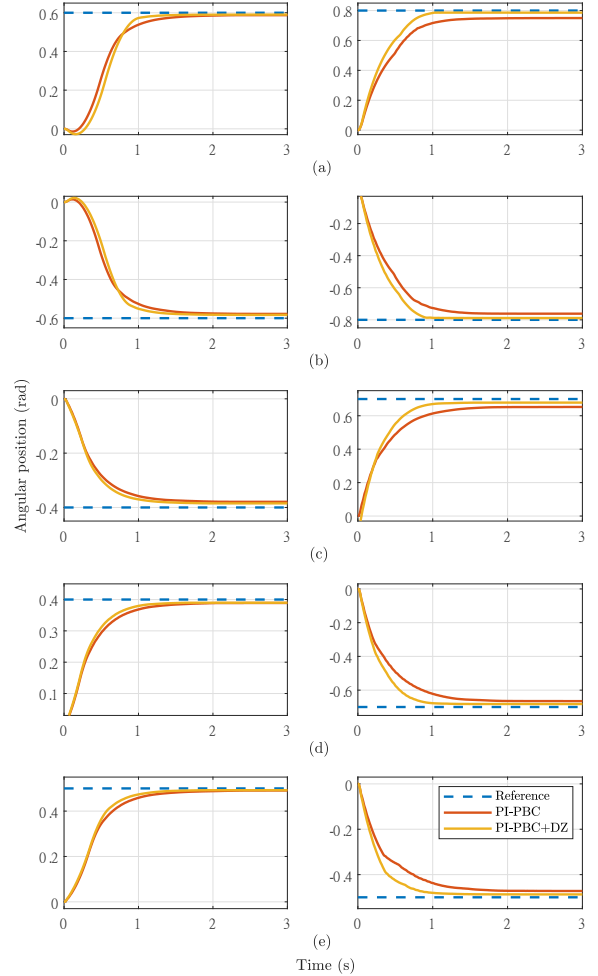}
	\caption{\centering Left: Angular trajectories for link 1 (L1). Right:  Angular trajectories for link 2 (L2).}\label{res}
\end{figure}

However, note that a small steady-state error remains with the proposed controller. This behaviour stems from assuming that the dead-zone is symmetric during the gain selection process, i.e.,  adjusting $\beta=0_n$. Then, to further reduce the steady-state error,  we can modify either $k_{zi}$ or $\beta$ from \eqref{ucompdz}. To show the behaviour of the closed loop for the latter scenarios, we stabilize the experiment setup at $q_\star=col(0.6,0.8)$ with three set of gains as shown in Table \ref{gain2}. We employ Case I as the baseline for comparison purposes, and the result for each case are shown in Table \ref{steady2} and Fig.~\ref{res2}. 
\begin{table}[h!]
	\caption{\centering Gains for $u_{\tt{pidz}}$}\label{gain2}
	\begin{tabular}{cccc}
		\hline
		& Case I            & Case II       & Case III           \\ \hline
		$K_P$   & \multicolumn{3}{c}{diag\{1.5,1\}}                      \\
		$K_I$   & \multicolumn{3}{c}{diag\{5,3\}}                        \\
		$K_Z$   & diag\{0.13,0.35\} & diag\{0.7,1\} & diag\{0.13,0.35\}  \\
		$\mu$   & 10                & 10            & 10                 \\
		$\beta$ & $0_2$             & $0_2$         & $col(-0.016,-0.2)$ \\ \hline
	\end{tabular}
\end{table}
\begin{table}[h!]
	\caption{\centering Steady-state error for each case}\label{steady2}
	\centering
	\begin{tabular}{ccc}
		\hline
		Case & \% Link 1 & \% Link 2 \\ \hline
		I    & 1.67      & 1.75      \\
		II   & 0.25      & 0.80      \\
		III  & 0.22      & 0.47      \\ \hline
	\end{tabular}
\end{table}

For both cases II and III, the steady-state error is clearly improved by augmenting $k_{zi}$ or adjusting $\beta$, respectively. We remark that by augmenting $k_{zi}$ in Case II, the steady-state error is reduced with respect to Case I at the expense of oscillations in the transient response. The overshoot stems from the overestimation of $k_{zi}$ and consequently, the condition \eqref{tuningrule} is not satisfied. Note that ${4\maxlamb{\mathcal{P}}\maxlamb{M_\star}=9.9883}$, then it follows that $\mathcal{R}$ must be chosen such that 
$\minlamb{\mathcal{R}}^2\geq 9.9883,$
to remove the overshoot. 
\begin{figure}[h!]
	\centering
	\includegraphics[width=0.85\columnwidth]{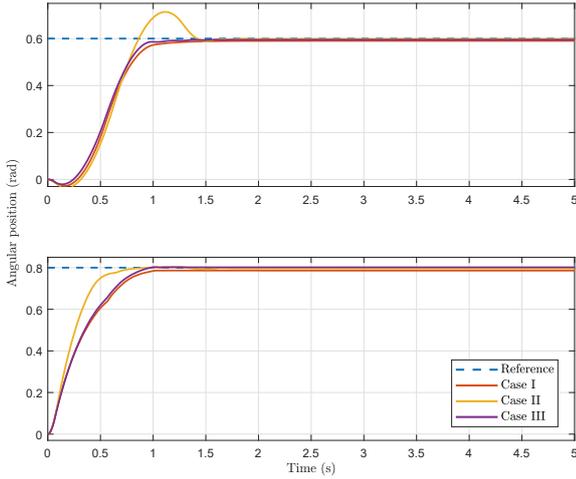}
	\caption{\centering \textbf{(Top)} Angular position for link 1. \textbf{(Bottom)} Angular position for link 2. }\label{res2}
\end{figure}

A video of the latter experiment can be found in:
{\tt{https://www.youtube.com/watch?v=774fufOyXpU}}.

\section{Concluding remarks and future works}\label{conclusions}
We have presented a PBC strategy that compensates for the dead-zone that is pervasive in servomechanisms, and we have provided a discussion on the performance of the controller near the equilibrium. Also, we have successfully implemented the proposed controller in an experimental setup where we have reduced the steady-state error compared to the standard PI-PBC and shown how to select the associated gains properly.  

Regarding possible future research, we propose to extend the described control strategy to underactuated mechanical systems. Furthermore, we aim to generalize the controller to include a more general dead-zone linearity, i.e., instead of \eqref{dzcharacteristic}, assume 
\cc{\begin{equation*}
	\tau_i=\begin{cases}
		g_1(v_i)(v_i-r_{\tt{bi}}),\qquad &v\geq r_{\tt{bi}}\\
		0,\qquad &\ell_{\tt{bi}}<v_i<r_{\tt{bi}}\\
		g_2(v_i) (v_i-\ell_{\tt{bi}}),\qquad &v\leq \ell_{\tt{bi}}
	\end{cases}
\end{equation*}}
for some $g_1(v_i),g_2(v_i)$ functions.
We also aim to extend the proposed approach with an adaptive control technique for scenarios in which the dead-zone is not available for measurement. 

\bibliography{ifacconf}            
\end{document}